\documentclass[11pt,a4paper]{article}
\usepackage[utf8]{inputenc}
\usepackage{authblk}
\usepackage{amsfonts}
\usepackage{amscd}
\usepackage{amsmath}
\usepackage{theorem}
\usepackage{mathrsfs}
\usepackage{marvosym}
\usepackage{fancybox,amssymb}
\usepackage{geometry}
\usepackage[english]{babel}
\usepackage{graphicx}
\usepackage{enumerate}
\usepackage{xcolor}
\usepackage{epstopdf}

\newcommand{\pia}{\pi_{\mathrm{all}}}
\newcommand{\pii}{\pi_{\mathrm{ind}}}

\newcommand{\mE}{\mathbb{E}}
\newcommand{\E}{\mathbb{E}}
\newcommand{\md}{\,{\rm d}}

\newcommand{\s}{\sum\limits}

\newcommand{\one}{1\mkern-5mu{\hbox{\rm I}}}

\theoremstyle{break}
\theorembodyfont{}
\newtheorem{Def}{Definition}[section]
\newtheorem{Rem}[Def]{Remark}
\newtheorem{Bem}[Def]{Remark}

\newtheorem{Lem}[Def]{Lemma}
\newtheorem{Prop}[Def]{Proposition}

\newtheorem{Thm}[Def]{Theorem}
\newtheorem{Bsp}[Def]{Example}

\makeatletter
\newenvironment{proof}{\noindent{\textit{Proof:}}}{%
\unskip\nobreak\hfil\penalty50\hskip1em\null\nobreak
$\Box$
\parfillskip=\z@\finalhyphendemerits=0\endgraf\bigskip}

\let\oldendBsp\endBsp
\def\endBsp{\unskip\nobreak\hfil\penalty50\hskip1em\null\nobreak\hfil%
$\blacksquare$\parfillskip=\z@\finalhyphendemerits=0\endgraf\oldendBsp}
\let\oldendBem\endBem
\def\endBem{\unskip\nobreak\hfil\penalty50\hskip1em\null\nobreak\hfil%
$\blacksquare$\parfillskip=\z@\finalhyphendemerits=0\endgraf\oldendBem}
\makeatother
\date{}

\title{Maximising with-profit pensions without guarantees}
\author{M. Carmen Boado-Penas\footnote{\Letter\;  carmenbo@liverpool.ac.uk}\quad\quad Julia Eisenberg\footnote{ jeisenbe@fam.tuwien.ac.at}\quad\quad Paul Kr\"uhner \footnote{peisenbe@liverpool.ac.uk}\\\vspace{-0.3cm}\footnotesize \hspace{0.7cm} University of Liverpool\quad\quad\quad\quad\quad\;  TU Wien \quad\quad \quad\quad University of Liverpool}
\normalsize
\begin{document} 
\maketitle
\begin{abstract}
\noindent
Currently, pension providers are running into trouble mainly due to the ultra-low interest rates and the guarantees associated to some pension benefits. With the aim of reducing the pension volatility and providing adequate pension levels with no guarantees, we carry out mathematical analysis of a new pension design in the accumulation phase. The individual's premium is split into the individual and collective part and invested in funds. In times when the return from the individual fund exits a predefined corridor, a certain number of units is transferred to or from the collective account smoothing in this way the volatility of the individual fund. The target is to maximise the total accumulated capital, consisting of the individual account and a portion of the collective account due to a so-called redistribution index, at retirement by controlling the corridor width. We also discuss the necessary and sufficient conditions that have to be put on the redistribution index in order to avoid arbitrage opportunities for contributors. \\\vspace{6pt}
\settowidth\labelwidth{{\it 2010 Mathematical Subject Classification: }}%
                \par\noindent {\it 2010 Mathematical Subject Classification: }%
                \rlap{Primary}\phantom{Secondary}
                93E20\newline\null\hskip\labelwidth
                Secondary 91B30, 91B08, 90B50 
\end{abstract}
\section{Introduction}
Pensions are in constant flux as insurers need to reinvent their products in an environment with continuous increases in longevity and ultra-low interest rates. At the same time employees desire security in retirement in the sense that they could get the retirement income they expect due to their past and current contributions into a pension scheme.\\
With-profits contracts (or participating policies in the US) were historically a significant part of the UK life insurance product palette. With-profits contract generally consists of a benefit if the individual dies within the term (term insurance) and a lump sum (pure endowment) if the policyholder survives within the term. This allows the policyholder to build up funds for a specific purpose such as an income in retirement. The important difference in this type of contracts is that additional periodic return can be given to the policyholder. In order to remove the short-term volatility of policyholder's payout value different smoothing mechanisms are applied in practice.\footnote{See Goecke (2013) and Guill\'en et al. (2016).} A with-profits investment\footnote{See Gatzert and Schmeiser (2013).} can either be conventional or unitised with the latter buying units in the with-profits fund. In the past, with-profits contracts often contained guarantees, like for example minimum guaranteed return, which allowed just for low-risk investments resulting in a lower expected value of the final accumulated amount.\\
With the aim of meeting consumer's needs in terms of stability after retirement, the dynamic hybrid life insurance\footnote{For more details on dynamic hybrid products see Bohnert et al (2014).} offers guarantees achieved by a periodical rebalancing process between three funds (the policy reserves, a guarantee fund and and equity fund). However, the investments are still made on low and average-risk products.\\
Currently, under the ultra-low interest rate economic environment, which significantly reduces the long-term benefits, the insurers try to avoid guarantees associated to pension products. Over the past few decades, in occupational pensions, traditional defined benefit (DB) plans are gradually losing their dominance and there has been a shift towards defined contribution (DC) pensions, where the investment risk is completely shifted from the insurer to the clients. Under a DC scheme the level of the pension is uncertain and, in general, without higher contribution rates will not produce decent benefits.\\
For workplace and private pension plans, collective defined contribution (CDC) schemes offer a middle ground between DB and DC plans. Under CDC, contributions are pooled and managed on a collective basis, and members own a proportional share of the aggregated collective investment rather than individuals share of the underlying assets as for the case of individual DC.\footnote{Collective pension schemes are also the dominant form of saving for retirement in countries such as the Netherlands and Norway – countries recognised as having among the best pension systems in the world according the Melbourne Mercer Global Pension Index (2019). See Bovenberg (2009), Hoevenaars and Ponds (2008), Ponds and Van Riel (2009) and Binsbergen et al. (2014), amongst others. In the UK, the Pension Scheme Act (2015) sets up a new legislate framework for private pensions encouraging shared risk pension schemes and collective benefits.} The plan also has a target pension amount – rather than a contractual guarantee-based on a long term and mixed risk investment plan. The way CDC adjusts the level of current and prospective pensions mean that there is an element of cushioning (smoothing) of volatility and much better long-term protection as the risk is shared by the members. This is because investment risk is adjusted over time and longevity risk is pooled across the membership. However, CDC entails some significant challenges particularly regarding the communication of the benefit calculation to members\footnote{Members have no control over the attribution of losses and surpluses.}, complex governance decision for trustees and high running costs that are likely to make it suitable only for larger schemes.\\
There is a clear need of security in retirement, i.e. satisfactory and stable pension benefits but at the same time the pension providers do not want to offer guarantees under the current low-interest rate environment.\\
With the aim of reducing the pension volatility and providing satisfactory pension levels, in this paper, we analyse a new pension design in the accumulation phase where the individual's premium is split into two accounts: individual account and collective account. Similar to unitised with-profit products, the premia in both accounts are invested in funds (the same or different ones). Depending on the performance of the individual fund, some units are transferred from or to the collective fund. In this way, the collective account acts as a buffer in smoothing mechanism for individual accounts. At retirement age, the individual receives a lifelong pension (or a lump sum payment) linked to her individual account and a portion of the collective account according to a so-called redistribution index, a weight identifying the part of the collective account belonging to a client according to her premium payment evolution. This paper builds upon Boado-Penas et al. (2019) and addresses the mathematical aspects to determine the optimal corridor for the exchange of units between the individual and collective accounts so that the final accumulated capital at retirement age is maximised. As a second objective, this paper discusses necessary and sufficient conditions for modelling the redistribution index.  Indeed, a thoughtless choice of the model might result in an arbitrage opportunity for some members of the pool of contributors.\\
Following this introduction, the next section of the paper describes the proposed pension model in the accumulation phase. Section 3 describes the mathematical model of the product and defines the target functional to maximise. Depending on the chosen help/gain-sharing procedures the optimal strategy for individual accounts has a different structure. First, we analyse a relatively simple case when the collective fund can never be empty, which corresponds to the full guarantee case. Second, we assume that it is not possible to get any help from the collective account if the total number of help units required by individual accounts exceeds the number of units in the collective account. Finally, we propose to use the redistribution index in order to specify the amount of help the individual accounts are entitled to require if the collective account cannot cover all accumulated individual claims. Section 4 provides a theoretical discussion on choice of a model for the redistribution index. Section 5 concludes and make suggestions for further research.
\section{The Model}
This section presents the mathematical formulation to determine the optimal corridor for the exchange of units between the individual and the collective account so that the total saved amount for the individual at retirement age is maximised. As optimisation criteria, we discuss the reasonability and mathematically feasibility of the optimal mean and optimal mean-variance.\\
For simplicity, we assume that the individual and collective funds are modelled by the same Geometric Brownian motion, $H_t$, where
\begin{equation*}
H_t=e^{x+\mu t+\sigma W_t}
\end{equation*}
with $W_t$ being a standard Brownian motion and the return of the funds expressed as
\[
\rho_t := \frac{H_t}{H_{t-1}}-1\;.
\]

\noindent We denote by $k$ and $-k$ the corridor boundaries of the individual fund, and $k\in[0,1]$. The returns exceeding the upper corridor boundary are partially distributed from the individual funds to the collective fund while the losses (negative returns falling out the lower corridor boundary) are partially compensated from the collective fund.

\noindent Let denote by $V^j_t$ the value of the $j$th individual account at time $t$, by $\eta^j_t$ the number of shares that belongs to the $j_t$ individual at time $t$ and by $C_t$ the value of the collective account at time $t$ with number of shares $\theta_t$.\\
The mathematical formulation of the with-profit procedure is as follows: 
\begin{itemize}
\item If $\rho_t > k$, then we say that the fund overperformed and a fraction of the surplus is transferred from the individual to the collective account
\[
\frac14 \Big(H_t-H_{t-1}(1+k)\Big)\eta_{t-1}=\frac14 V_{t-1}\Big(\frac{H_t}{H_{t-1}}-1-k\Big)=
\frac14 V_{t-1}\big(\rho_t-k\big),
\] 
i.e.\ one transfers $\frac14\Big(1-\frac{H_{t-1}}{H_t}(1+k)\Big)\eta_{t-1}$ units of $H_t$ into the collective account.
\item If $\rho_t < -k$, then we say that the fund overperformed and in this case the individual account creates a claim from the collective account of
\[
\frac12 \Big(H_{t-1}(1-k)-H_t\Big)\eta_{t-1}=\frac12 V_{t-1}\Big(1-k-\frac{H_t}{H_{t-1}}\Big)=\frac12 V_{t-1}\big(-k-\rho_t\big)\;,
\] 
i.e $\frac12\Big(\frac{H_{t-1}}{H_t}(1-k)-1\Big)\eta_{t-1}$ units will be transferred into the individual account.
\end{itemize}

 %The effect of the boundary $k$ 

 \noindent The corridor boundaries make the funds change smoother. 
 %is that over- and underperformaces will be less severe as one gets compensated for underperformance and punished for overperformance. 
 It is clear that, if compared to the case with no barriers, this will reduce the realised volatility.\\
The transactions described above involve units of the collective fund in and out. This mechanism faces a problem if the collective account is plundered by individual accounts too often. It would mean that the individual accounts are not profitable and need a continuous support. For this reason the choice of barrier is restricted to those $k\in[0,1]$ such that the collective account does not loose money in expectation. This leads to the following profitability condition:\bigskip
\\\textbf{Profitability condition:} The set of admissible $k\in[0,1]$ is given by those $k$ such that 
\begin{align}
\mE\Big[\frac 12 \Big(1-k-\frac{H_t}{H_{t-1}}\Big)^+-\frac 14 \Big(\frac{H_t}{H_{t-1}}-1-k\Big)^+ \Big]\le 0\;.\label{profit}
\end{align} 
%The reason for restricting the set of the admissible $k$ is that the collective fund should not ruin almost surely due to the withdrawals from the individual funds. 
The idea of profitability is not new. For instance, in ruin theory the net profit condition, that states that the expected total loss should be strictly smaller than the expected earnings from premia payments, is required.\footnote{See Dickson (2005).}
\\At the end of the accumulation phase, at the retirement point, the total saved amount will consist of two parts: the total saved amount from the individual account and a part from the collective account. The part from the collective account is calculated due to a so-called redistribution index, which we denote by $J^j_{T-1}$ for the $j$th contract. In order to prevent arbitrage, at time $T$ we use the redistribution index determined at time $T-1$. 
%Further, we define a redistribution index, $J(t)$, describing the part of the collective fund belonging to an individual under consideration at time $t$ and it is set at time $t-1$. 
The index is only updated during the accumulation phase when new contributions (premia payments) are made and constant otherwise which means that it is discrete in nature. %We do not make a detailed description on the update-mechanism for $J$ here but simply assume that there is some ``fair'' mechanism which we cannot influence. 
For a more detailed discussion on choices for the index see Section \ref{reindex}. %The primary target of the index as already mentioned above is to determine the collective fortune transferred to an individual when leaving the pool of contributors.  
\\
The redistribution index might also be needed in order to determine the number of units to be transferred into an individual account in case of underperformance if the collective account does not have enough units to cover all occurred claims. This is critical as there is a lack of analysis in both academic research and regulation with respect to the strategy the insurance company should adopt in this case. We intend to fill in this gap through some possible scenarios in the following section.
%latter situation is critical as it is not clear from the regulatory nor from the optimisation point of view which strategy an insurance company should adopt in this case. 
%Some possible scenarios are considered in the following section.
\section{Maximisation of the total saved capital}
From our model setup it is clear that the total saved capital will depend on the transactions between the individual and the collective accounts as explained in the previous sections. Therefore, the way the transactions are performed will impact the optimisation procedures. In this section we consider three scenarios. First, we assume that the collective account always has enough units to cover all claims from the individual account. %This corresponds to a guarantee case when the insurance company or a third party undertakes the coverage of the possible losses. 
In the second scenario, in case of an insufficient number of units in the collective account to cover individual claims, no units are transferred from the collective fund. Third, we use the redistribution index $J$ in order to specify the number of units that can be transferred to an individual despite the deficit in the collective account.
\subsection{The collective fund is never empty \label{notempty}}
In this section we show the mathematical aspects of the maximisation of the total capital at retirement point assuming that the collective account always have sufficient number of units to be transferred into the individual account.This is a huge mathematical simplification but if it can be ensured that the collective account never ruins, then this is, in fact, satisfied.\\
Following a Markovian structure and the saved amount in every single period is optimised independently from the past.\\
In order to find the optimal choice of the boundary $k_{t-1}$ for the period $[t-1,t]$, we will make use of a recursive backward search, starting our considerations in the period $[T-1,T]$. 
\\The wealth in the individual account of a policyholder consists of the part $\gamma\in[0,1]$ of the premia $\pii$ paid to it, the investment returns, and the amount transferred to or from the collective account.
\begin{align*}
V_t&=\gamma\pii +\eta_{t-1}H_t-\one_{[\rho_t>k]}\frac14 V_{t-1}\Big(\rho_t-k\Big)
\\&\quad{}+ \one_{[\rho_t<-k]}\frac12 V_{t-1}\Big(-k-\rho_t\Big)\;.
%\gamma\pi +\eta_{t-1}H_t-\one_{[\frac{H_t}{H_{t-1}}>1+k]}\frac14 V_{t-1}\Big(\frac{H_t}{H_{t-1}}-1-k\Big)
%\\&\quad{}+ \one_{[\frac{H_t}{H_{t-1}}<1-k]}\frac12 V_{t-1}\Big(1-k-\frac{H_t}{H_{t-1}}\Big)\;.
\end{align*}
The wealth of the $j$th individual is denoted by $V^j$ (or simply by $V$ in case we refer to a representative individual) and the wealth of the collective fund is described by its investment returns, the part $1-\gamma$ of all the premia $\pia$, and the gains or losses from the transactions with all the individual accounts
\begin{align*}
C_t&=(1-\gamma)\pia +\theta_{t-1}H_t+
\\&\quad{}\sum_{j} \Bigg(\one_{[\rho_t > k]}\frac14 V^j_{t-1}\Big(\rho_t-k^j\Big)- \one_{[\rho_t<-k^j]}\frac12 V^j_{t-1}\Big(-k^j-\rho_t\Big)\Bigg)
%C_t&=(1-\gamma)\pi +\theta(t-1)H_t+\one_{[\frac{H_t}{H_{t-1}}>1+k]}\frac14 C(t-1)\Big(\frac{H_t}{H_{t-1}}-1-k\Big)
%\\&\quad{}- \one_{[\frac{H_t}{H_{t-1}}<1-k]}\frac12 C(t-1)\Big(1-k-\frac{H_t}{H_{t-1}}\Big)
%\\&= (1-\gamma)\pi +\theta(t-1)H_t+\one_{[\frac{H_t}{H_{t-1}}>1+k]}\frac14 \eta(t-1)\Big(1-(1+k)\frac{H_{t-1}}{H_t}\Big)H_t
%\\&\quad{}- \one_{[\frac{H_t}{H_{t-1}}<1-k]}\frac12 \eta(t-1)\Big((1-k)\frac{H_{t-1}}{H_t}-1\Big)H_t
\end{align*}

Our target criterion for an individual is given by
\[
 A(k_1,\dots,k_T) := \mE\big[V_T\big]\;%+J(T)C_t\big]\;
\]
which is to be maximised over all possible choices of boundaries $k_1,\dots,k_T$ at every point of time for every individual\footnote{Optimal boundaries will be chosen by the insurance company for the individuals and not by the individuals themselves.} where the boundary $k_t$ is decided at time $t-1$, i.e.\ $\mathcal F_{t-1}$-measurable. \\
Due to the assumption that the collective account does not ruin, the choices for the boundaries are not influenced by the choices for other individuals. In order to show this we define the function
 \begin{align*}
\Psi_1(k) :&= \mE\Big[\rho_T-\frac{1}{4}(\rho_T-k)^++\frac12(-\rho_T-k)^+\Big]
            \\ &= \mE\Big[\frac{V_T - V_{t-1}-\gamma\pii}{V_{t-1}}\Big]
 \end{align*}
 for $k\in[0,1]$. We make an observation regarding the maximum of $\Psi_1$ first.
\begin{Lem}\label{lem:1}
Define for arbitrary $1<a<b$
\begin{align}
\Xi(k):=\mE\Big[\rho_t+\frac1a \Big(-\rho_t-k\Big)^+-\frac1b \Big(\rho_t-k\Big)^+\Big], \label{xi}
\end{align}
then the maximum is attained either at $k=1$ or at the minimal $k$ allowed by the profitability condition. If the profitability condition is not assumed, then the maximum is attained in either $0$ or $1$.
\end{Lem}
\begin{proof}
%In order to prove our  claim we will consider the derivatives of the function $\Xi$ defined above.
Let $f$ denote the density of the random variable $\frac{H_T}{H_{T-1}}$, which is given by
\begin{align}
f(y)= \frac 1{\sqrt{2\pi}y\sigma}e^{-\frac{(\ln(y)-\mu)^2}{2\sigma^2}}\;. \label{density}
\end{align}
because $H$ is a geometric Brownian motion. The derivatives of $\Xi$ are given by
\begin{align*}
&\Xi'(k)=-\frac 1a\int_0^{1-k} f(y)\md y+ \frac 1b \int_{1+k}^\infty f(y)\md y,
\\&\Xi''(k)=\frac 1a f(1-k)-\frac 1b f(1+k)\;.
\end{align*}
therefore
%One can immediately conclude that
\begin{align*}
&\Xi'(1)=\frac 1b \int_{2}^\infty f(y)\md y>0,
\\&\Xi''(0)=\Big(\frac1a-\frac 1b\Big)f(1)>0,
\\&\Xi''(1)=-\frac 1b f(2)<0\;.
\end{align*}
Assume now $k^0:=\inf\{k\in[0,1]: \; \Xi''(k)< 0\}$, i.e. $\Xi''(k^0)=0$ and $\Xi''(k)\ge 0$ for $k\le k^0$. Inserting $\frac1a f(1-k^0)=\frac 1b f(1+k^0)$ into $\Xi'''$ yields
\[
\Xi'''(k^0)=\frac{f(1-k^0)}{\sigma(1-k^0)(1+k^0)}\Big[2\sigma -2\mu+(1+k^0)\ln(1-k^0)+(1-k^0)\ln(1+k^0)\Big].
\] 
The expression in quadratic brackets above is strictly decreasing in $k$, converging to $-\infty$ if $k$ approaches $1$, and has one zero point. Since $\Xi''$ needs to change the sign in order to be negative at $k=1$ and $k^0$ is a zero point, it must hold $\Xi'''(k^0)\le 0$ with $\Xi'''(k)<0$ for all $k>k^0$. It means that $\Xi''$ after becoming negative once, at $k^0$ will stay negative until $k=1$. Therefore, we can conclude for $\Xi'$ that it has a unique local maximum at $k^0$, it is strictly increasing before $k^0$ and strictly decreasing thereafter. Since $\Xi'(1)>0$ we find that it has at most one zero, i.e.\ it is strictly positive after its zero and stricly negative before its zero or alternatively, strictly positive everywhere. Thus, $\Xi$ is strictly decreasing until it reaches its minimum at the zero of $\Xi'$ and strictly increasing thereafter or alternative $\Xi$ is strictly increasing everywhere. In the latter case its maximum is attained at $k=1$ and in the former case the maximum is either attained in $k=1$ or in the minimal $k$ allowed by the profitability condition.
\end{proof}
\begin{Prop}
The optimal choice for the boundary $k$ is given by the maximiser $k$ of the function $\Psi_1$ and it does neither depend on the time nor the individual. That is, choosing $k_1,\dots,k_T$ such that they are equal and $k_1$ maximises $\Psi_1$ yields
   $$ \sup_{\tilde k_1,\dots \tilde k_T}A(\tilde k_1,\dots,\tilde k_T) = A(k_1,\dots,k_1). $$
Moreover, $k_1$ is either the maximal or minimal allowed value.
\end{Prop}
\begin{proof}
Lemma \ref{lem:1} yields the additional statement.
\\We have
  \begin{align*}
     A(k_1,\dots,k_T) &= \sum_{t=1}^T \mE[\Delta V_t]
\end{align*}    
where $\Delta X_t:= X_t-X_{t-1}$ for any process $X$ and $t\geq 1$. We will see that there is a choice of boundaries which maximises each summand and, hence, maximises the sum. We have
  $$ \Delta V_t = \gamma\pii + V_{t-1}\frac{V_t-V_{t-1}-\gamma\pii}{V_{t-1}} $$
  which yields by the tower property
  \begin{align*}
     \mE[\Delta V_t] &= \gamma\pii +\mE\left[V_{t-1}\mE\left[\frac{V_t - V_{t-1}-\gamma\pii}{V_{t-1}}\Big|\mathcal F_{t-1}\right]\right] \\
      &= \gamma\pii +\mE\left[V_{t-1}\Psi_1(k_t)\right]
  \end{align*}
for any $t=1,\dots, T$ where the later equality holds due to the i.i.d.\ property of the returns. Since $V_{t-1}$ is positive we find that the maximiser $k^0$ of $\Psi_1$ is the optimal choice for $k_t$ when maximising $\mE[\Delta V_t]$. Thus, we have
   $$ \sup_{k_1,\dots,k_T}\mE[\Delta V_t] = \gamma\pii +\Psi_1(k^0)\sup_{k_1,\dots,k_{t-1}}\mE\left[V_{t-1}\right]. $$
Consequently, we find that $k_t = k^0$ is the optimal choice. Lemma \ref{lem:1} yields that $k^0\in \{0,1\}$.
 \end{proof}
As we have seen in the proof above, by looking at the terms depending on $k$ is sufficient to maximise
%it suffices to look just at therms depending on $k$. It suffices to maximise
\begin{align}
M_1(k):=\{1-J_{T-1}\}\mE\Big[\frac12 \Big(1-k-\frac{H_T}{H_{T-1}}\Big)^+-\frac14 \Big(\frac{H_T}{H_{T-1}}-1-k\Big)^+\Big]\;.\label{M1}
\end{align}
\begin{Bsp}[Asymmetric boundaries]
\begin{figure}[t]
\includegraphics[scale=0.4, bb = -200 0 500 520]{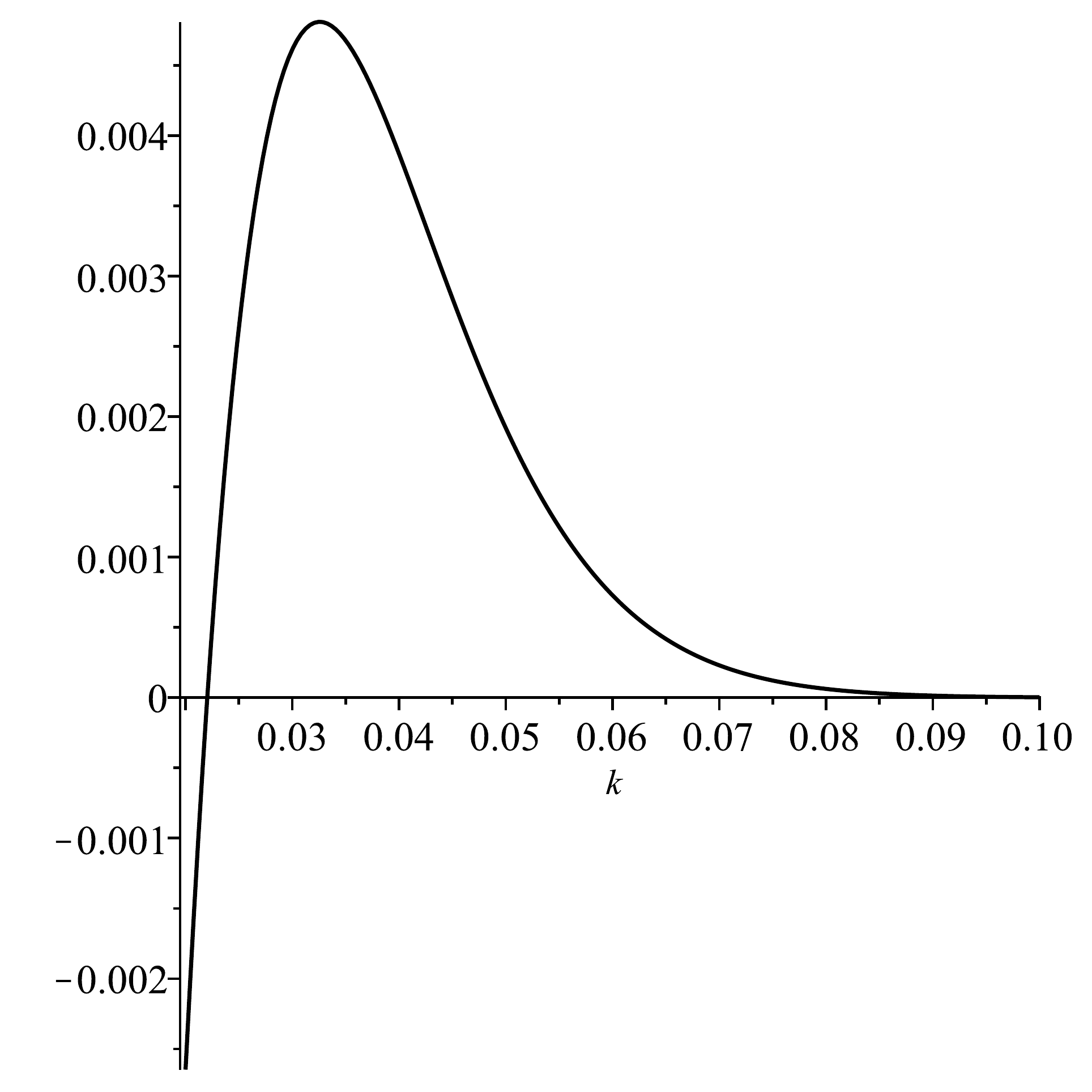}
\caption{The function $M_1(x)$.\label{2}}
\end{figure} 
If we drop the assumption of the symmetric boundaries 
and assume for instance $-k$ for the lower and $k\cdot p$ for the upper boundary with some given $p>1$, we can get the following result. Let $\mu:= 0.015$, $\sigma:=0.03$ and $p=2$, then maximizing the function
\[
M_p(k):=\{1-J_{T-1}\}\mE\Big[\frac12 \Big(1-k-\frac{H_T}{H_{T-1}}\Big)^+-\frac14 \Big(\frac{H_T}{H_{T-1}}-1-kp\Big)^+\Big]
\]
yields $k= 0.03257$ and correspondingly $kp=0.06515$, confer Figure \ref{2}. i.e. the policyholders should get help from the collective fund if the individual fund goes down more than $3.3\%$, and transfer money into the collective fund if the individual fund goes up more than $6.5\%$. However, as shown in Figure \ref{2}, the profitability condition required in \eqref{profit} is not fulfilled. Thus, $k=0.03257$ is not an admissible strategy. 
\end{Bsp}
%The example above shows that assuming asymmetric boundaries does not lead to satisfactory results.
As our target is to smooth the evolution of individual portfolios we need some kind of penalty for high volatility in order to obtain a maximum. Such a penalty function can be the expected realised volatility of the fund. Here, one has to decide if a relative or an absolute value should be considered. We follow the relative ansatz. It means we optimise
 $$ A(k_1,\dots,k_T) := \mE[V_t] - \alpha \mE\left[ \s_{t=1}^T \frac{1}{V_{t-1}}(V_t-V_{t-1}-\gamma\pii)^2 \right] $$
 for some weight $\alpha >0$ where we optimise the boundaries $k_1,\dots,k_T$ at every point of time. As such an $\alpha$ the insurance company may choose the desired proportion between the mean and the realised variation. 

We observe the following identity and define two functions
   \begin{align*}
       \frac{V_t - V_{t-1}-\gamma\pii}{V_{t-1}} &= \rho_T+\frac12 \Big(-\rho_T-k\Big)^+-\frac14 \Big(\rho_T-k\Big)^+,
    \end{align*}
    \begin{align*}   
       \Psi_1(k) &:= \mE\Big[\frac{V_t - V_{t-1}-\gamma\pii}{V_{t-1}}\Big], \\
       \Psi_2(k) &:= \mE\Big[\Big(\frac{V_t - V_{t-1}-\gamma\pii}{V_{t-1}}\Big)^2\Big],
   \end{align*}
 for any $k\in[0,1]$ which allow to simplify the optimisation problem.
\begin{Lem}
 We have
  $$ A(k_1,\dots,k_T) = \sum_{t=1}^T \mE\big[V_{t-1}\left(\Psi_1(k_t)-\alpha \Psi_2(k_t)\right)\big] $$
  for any choice of boundaries $k_1,\dots,k_T$. In particular, it is optimal to choose $k_1,\dots,k_T$ equal such that the expression 
   $$ \Psi_1(k_1)-\alpha \Psi_2(k_1) $$
   is maximised.
\end{Lem}
\begin{proof}
We have 
\begin{align*}
A(k_1,\dots,k_T) &= \mE\bigg[\sum_{t=1}^T \Delta V_t\bigg]\\
      & =\sum_{t=1}^T \mE\Bigg[V_{t-1}\Bigg\{ \mE\bigg[\frac{V_t-V_{t-1}-\gamma\pii}{V_{t-1}}\bigg|\mathcal F_{t-1}\bigg] \\&\quad{}- \alpha \mE\bigg[ \left(\frac{V_t-V_{t-1}-\gamma\pii}{V_{t-1}}\right)^2\bigg|\mathcal F_{t-1} \bigg]\Bigg\} +\gamma\pii \Bigg] \\
      & = \sum_{t=1}^T \Big(\mE\big[V_{t-1}\big(\Psi_1(k_t)-\alpha \Psi_2(k_t)\big)\big]+\gamma\pii\Big).
    \end{align*}
This shows that an optimal choice for $k_T$ is a maximum of the function $\Psi_1-\alpha\Psi_2$. A simple induction shows that $k_1=\dots=k_T$ with the above choice of $k_T$ is optimal.
 \end{proof}
 \begin{figure}[t]
\includegraphics[scale= 0.5, bb = -100 0 300 350]{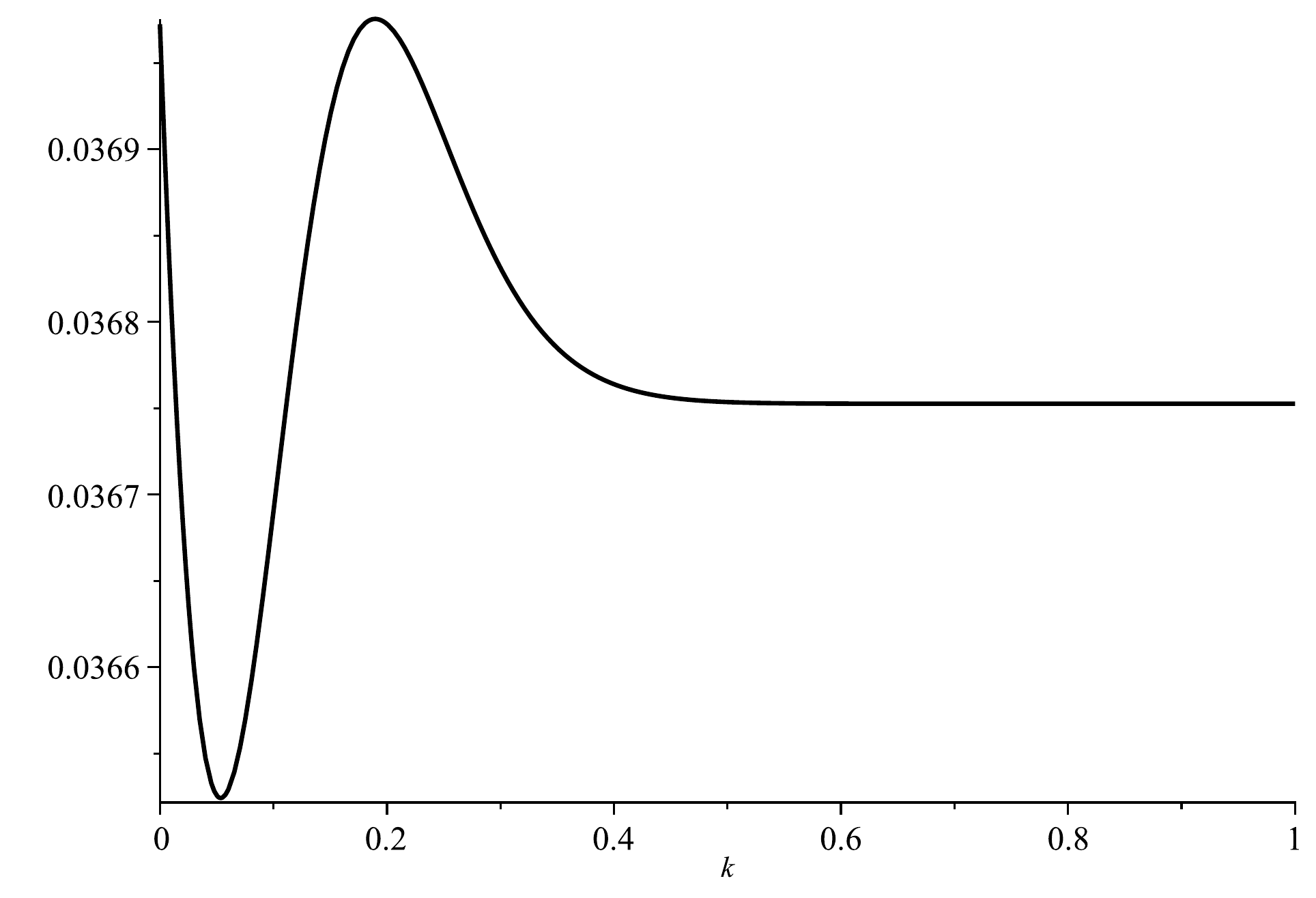}
\caption{The function $M_2(k)$ for $\sigma=0.092367$, $\mu= 0.06$, $\alpha= 2$ and $J=0.02$, yielding $M_2(0)=M_2(0.1897)$. The profitability condition \eqref{profit} is fulfilled on $[0,1]$.\label{fig:11}}
\end{figure}
From the preceding lemma we know that the optimal choice of boundary $k$ is the maximiser of the following (time-independent) functional:
\begin{align*}
M_2(k)&:=\Psi_1(k)-\alpha\Psi_2(k)%V_{t-1}\times
\\&=\bigg\{\mE\Big[\rho_T+\frac12 \Big(-\rho_T-k\Big)^+-\frac14 \Big(\rho_T-k\Big)^+\Big]
\\&\quad{}-\alpha\mE\Big[\Big(\rho_T+\frac12 \Big(-\rho_T-k\Big)^+-\frac14 \Big(\rho_T-k\Big)^+\Big)^2\Big]\bigg\}\;,
%\\&{}-\alpha V_{t-1}{\rm Var}\Big[\frac{H_T}{H_{T-1}}-1+\frac12 \Big(1-k-\frac{H_T}{H_{T-1}}\Big)^+-\frac14 \Big(\frac{H_T}{H_{T-1}}-1-k\Big)^+\Big]\;.
\end{align*}
%In the literature one usually assumes the weight $\alpha$ as already given. However, no explanations about the meaning of different $\alpha$-values are provided. By introducing a new product it is of crucial importance to reasonably justify every variable used. 
\begin{Rem}
  The function $\Psi_2$ is strictly increasing. Lemma \ref{lem:1} shows that $\Psi_1$ is first decreasing until reaching its minimum and increasing thereafter. Consequently, $\Psi_1-\alpha \Psi_2$ is first decreasing and may start to increase at a later time but this cannot be before the minimum of $\Psi_1$. Thus, the maximum of $\Psi_1-\alpha \Psi_2$ as a function on $[0,1]$ is either attained in $0$ or after the minimum of $\Psi_1$.
\\It might happen, confer Figure \ref{fig:11}, that the maximum of $\Psi_1-\alpha \Psi_2$ cannot be uniquely defined, i.e. the  set ${\rm argmax}\{\Psi_1-\alpha \Psi_2\}$ contains at least two elements, say $k_1<k_2$. Recall that maximising the value of an individual account corresponds to the maximisation of the function $M_1$ defined in \eqref{M1} and leads to a bang bang strategy. Since, the individual accounts yield the main basis for the calculation of the initial pension, we choose $k_1$ if $M_1$ is decreasing and $k_2$ if $M_1$ is increasing in order to optimise the value of individual accounts.
\end{Rem}
\begin{Bsp}
Let us again assume that $\mu=0.045$, $\sigma=0.06$ and $\alpha=4$. The penalised function $M_2(k)$ is given in Figure \ref{fig2}. The maximum is attained at $k=0.1215$. The profitability condition \eqref{profit} is fulfilled for all $k\in[0,1]$. 
\begin{figure}[t]
\begin{minipage}[h]{\textwidth}
\includegraphics[scale=0.3, bb = -30 0  300 540]{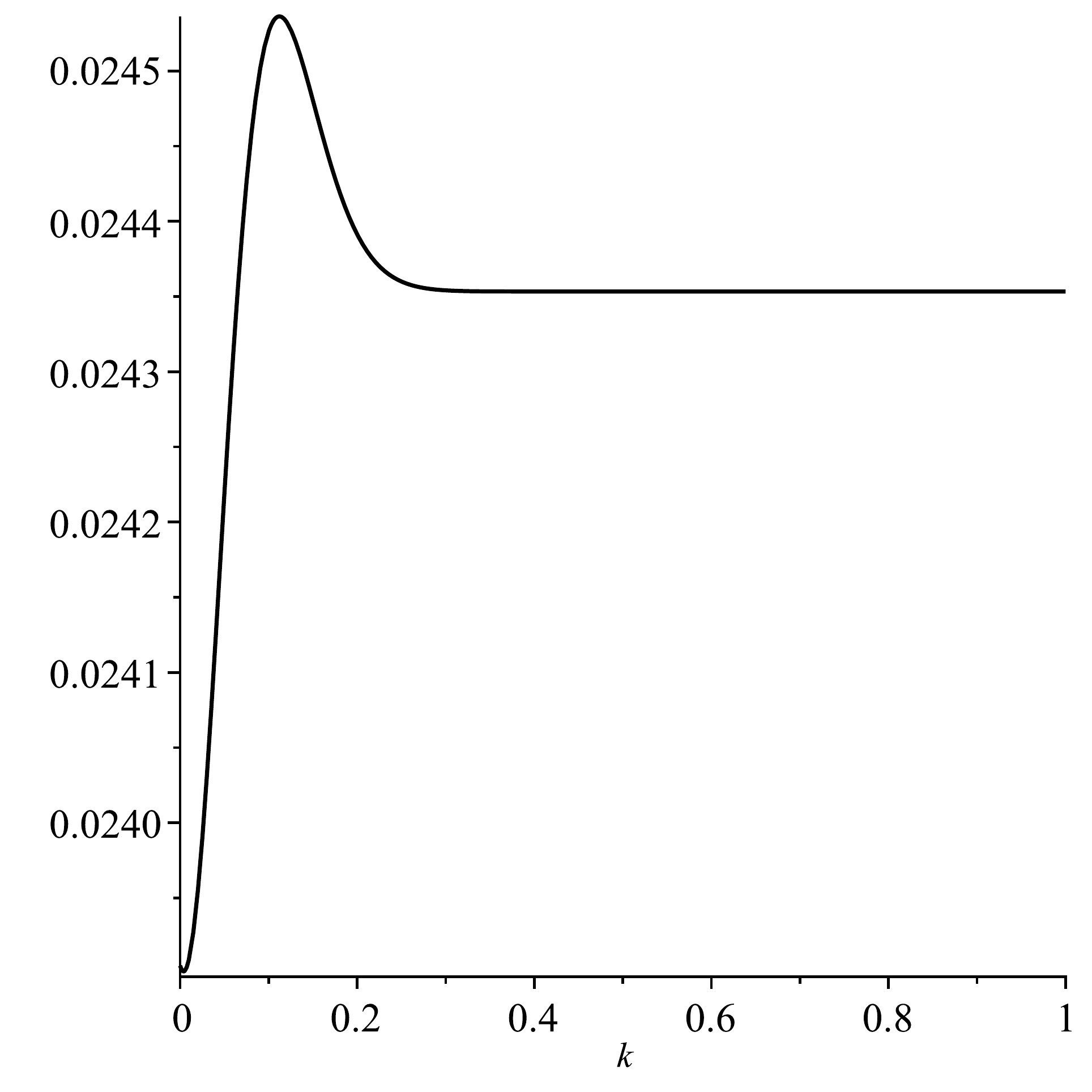}
\includegraphics[scale=0.3, bb = -320 0  450 540]{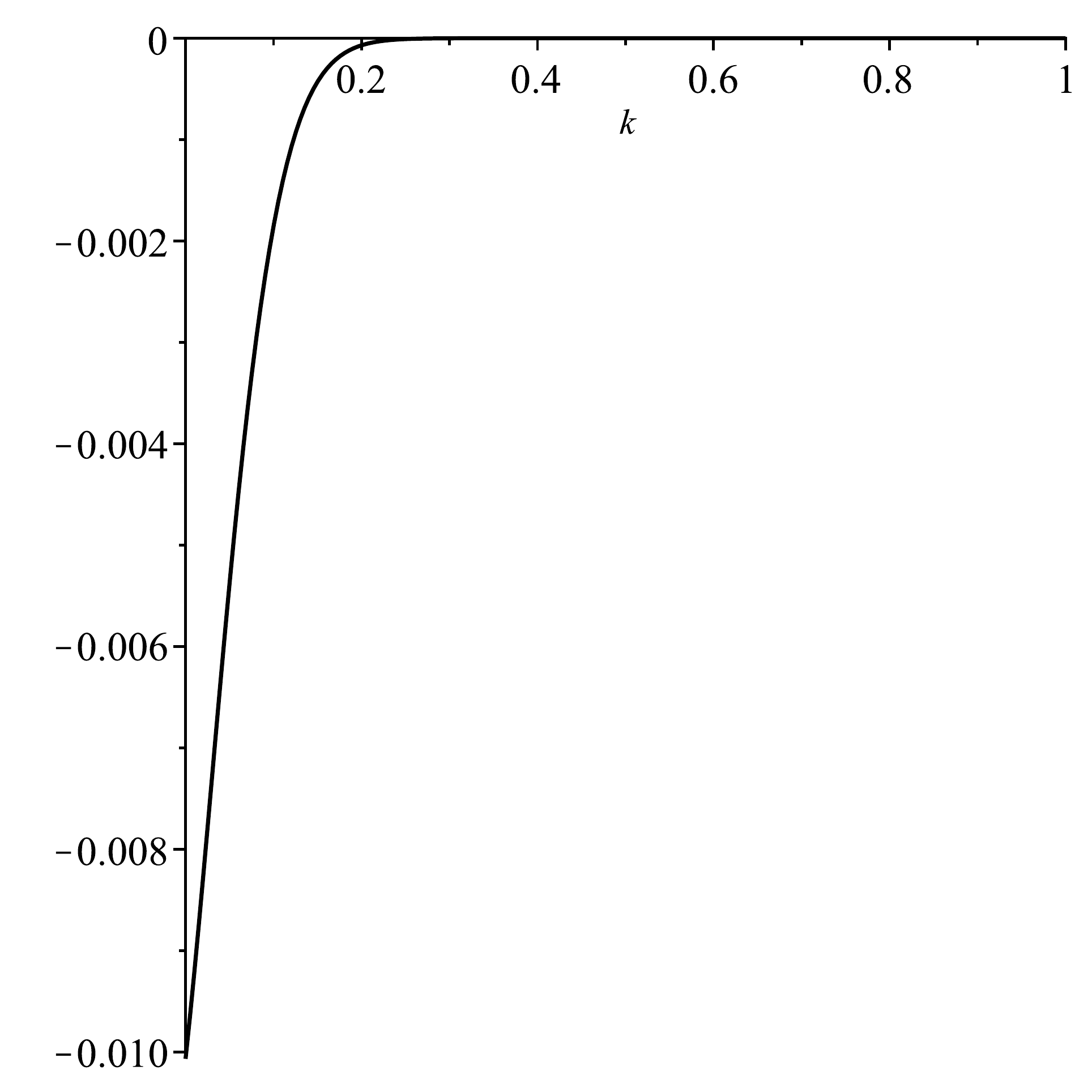}
\end{minipage}
\caption{$M_2(k)$ (left picture) and profitability condition (right picture).\label{fig2}}
\end{figure}
\end{Bsp}
\subsection{No help if the collective account does not have sufficient number of units} \label{sec2}
The problem of the assumption that the collective account can never become empty is the fact that neither insurance companies, nor the employer nor the state might be willing to cover the deficit if any. Therefore, in this section, we assume that in the case that the collective account does not have enough units in order to cover all claims at a particular point in time, no single claim will be paid. 
%This has implications for the transfer between individual and collective account if there is an insufficient amount in the collective. To capture this we put up a modified evolution for the wealth of an individual.
The evolution of the wealth for an individual $j$ is given by
\begin{align*}
V^j_t&=\gamma\pi_{\rm ind} +\eta_{t-1}H_t-\frac14 V^j_{t-1}\Big(\rho_t-k^j\Big)^+\nonumber
\\&\quad{}+ \frac12 V^j_{t-1}\Big(-\rho_t-k^j\Big)^+\one_{\big[2\theta_{t-1}(1+\rho_t)\geq \s_{i=1}^n\eta^i_{t-1}\big(-k^i-\rho_t\big)^+\big]}\;,
\end{align*}
where $n$ is the total number of the contracts in the insurance pool, $k^1,...,k^n$ and $\eta^1_{t-1},...,\eta^{n}_{t-1}$ are the boundaries and the number of shares in the individual accounts at that time respectively.\\ 
Note that here we index the boundaries by individuals rather by time but new thresholds can be chosen dynamically at discrete time points. The individual account under consideration is indexed by $j$. On the left hand side of the indicator we find twice the value of the collective account before any units are transferred and on the right hand side the total volume of all individual losses exceeding the individual thresholds. If the left hand side in the indicator is not bigger, then there is insufficient wealth to cover half of the individual excess losses. In that case, no one gets any help at all to prevent the collective fund to become empty.\\Our target is to optimise the expected return minus the relative realised quadratic variation for each individual, i.e.\ we aim at optimising
 $$ \mE[V^j_t] - \alpha \mE\left[ \sum_{t=1}^T \frac{1}{V^j_{t-1}}(V^j_t-V^j_{t-1}-\gamma\pii)^2 \right] $$
 for each individual. The problem here is the cross dependence among all individuals. 
 %We merely suggest a choice where each individual uses the same barrier but can improve only by very little. 
One possibility could be that all individuals use the same barrier, chosen by the insurance company. 
We make a precise error analysis in the sense that we single out how much an individual can improve and show that this depends only on the fraction of its wealth compared to the collective wealth, which in a large community should be rather small.\\ We try to find a common choice of barriers such that no individual has an improved target value if all barriers are increased or decreased a bit. \\Our main result of this section, Theorem \ref{t:diagonal} below, states that it is optimal to choose the same barrier for all individuals and dynamically increase the barriers if the amount in the collective account is relatively low compared to the total wealth of all individual accounts.

%From a game theoretic point of view, we are looking for a Nash equilibrium, i.e.\ a choice of barriers $k^1,\dots,k^n$ such that for each individual it is unpreferable to change the barriers chosen for it.

Following the same arguments as in the previous sections we can see that it is optimal to optimise at each time step separately, i.e.\ at time $t$ for each individual $j$ we need to optimise
 $$ \mE[U^j_t-\alpha (U^j_t)^2] $$
where
 $$ U^j_t = \rho_t - \frac14(\rho_t-k^j)^++\frac{1}{2}\left(-k^j-\rho_t\right)^+\one_{\big[2\theta_{t-1}(1+\rho_t)> \s_{i=1}^n\eta^i_{t-1}\big(-k^i-\rho_t\big)^+\big]} $$
 and $k^1,\dots,k^n$ are to be chosen $\mathcal F_{t-1}$-measurable. Our target criterion implies that an optimal choice of barriers is a boundary value or
  $$ \sum_{i=1}^n \partial_{i_j} \mE[ U^j_t-\alpha (U^j_t)^2 ] = 0 $$
 for each $j$.

\begin{Rem}\label{R:z star}
  If a (possibly non-optimal) choice of thresholds has been made, then the indicator, as a function of $\rho_t$ is decreasing and, hence, there is some constant $z^*(k^1,\dots,k^n)$ such that 
   $$ \one_{\big[2\theta_{t-1}(1+\rho_t)\geq  \s_{i=1}^n\eta^i_{t-1}\big(-k^i-\rho_t\big)^+\big]} = \one_{[\rho_t \geq z^*(k^1,\dots,k^n)]}.$$
   
  Also we have
   \begin{align*}
      z^*(k^1,\dots,k^n) &= -\frac{2\theta_{t-1}+\sum_{i\in I}\eta^i_{t-1}k^i}{2\theta_{t-1}+\sum_{i\in I}\eta^i_{t-1}} \in[-1,0] \\
      I &:= \left\{j=1,\dots,n : 2\theta_{t-1} (1-k^j)-\sum_{i=1}^n \eta^j_{t-1}(k^j-k^i)^+ \geq 0 \right\}
\end{align*}    
   
   Note that $z^*$ is Lipschitz-continuous and its absolutely continuous derivative is given by
    $$ \partial_{k^j}z^*(k^1,\dots,k^n) = \frac{-\eta^j_{t-1} \one_{\{j\in I\}}}{2\theta_{t-1}+\sum_{i\in I}\eta^i_{t-1}}.$$

  In the particular case that all $k^i$ are equal we find $I=\{1,\dots,n\}$ and, hence, the following simplifications
   \begin{align*}
      z^*(k^1,\dots,k^1) &= -\frac{2\theta_{t-1}+\sum_{i=1}^n\eta^i_{t-1}k^1}{2\theta_{t-1}+\sum_{i=1}^n\eta^i_{t-1}}, \\
      \partial_{k^j}z^*(k^1,\dots,k^1) &= \frac{-\eta^j_{t-1}}{2\theta_{t-1}+\sum_{i=1}^n\eta^i_{t-1}}.
   \end{align*}
    
This reveals that if the same barrier $k$ is chosen for all individuals and if the $j$-th individual has negligible amount compared to the total amount of all other individuals plus the collective amount, then the $\partial_{k^j}$-derivative is negligible as well.
\end{Rem}

%\begin{Prop}\label{p:jth ind}
%  We have
%    \begin{align*} \partial_{k^j}\E[U_t^j] &= \E\left[\frac14\one_{\{\rho_t>k^j\}}-\frac12\one_{\{\rho_t<-k^j\}}\one_{\{\rho_t>z^*(k^1,\dots,k^n\}}\right]
%    \\&\quad + \frac{1}{2}(-k^j-z^*(k^1,\dots,k^n))^+f(z^*(k^1,\dots,k^n))\partial_{k^j}z^*(k^1,\dots,k^n) \\
%    \partial_{k^j}\E[(U_t^j)^2] &= \E\left[U_t^j\left(\frac12\one_{\{\rho_t>k^j\}}-\one_{\{\rho_t<-k^j\}}\one_{\{\rho_t>z^*(k^1,\dots,k^n\}}\right)\right]
%    \\&\quad +[U_t^j|\rho_t=-z^*(k^1,\dots,k^n)](-k^j-z^*(k^1,\dots,k^n))^+f(z^*(k^1,\dots,k^n))\partial_{k^j}z^*(k^1,\dots,k^n)
%    \end{align*}
%  where $f$ is the density of $\rho_t$ and $[U_t^j|\rho_t=-z^*(k^1,\dots,k^n)]$ is the value of $U_t^j$ if $\rho_t=-z^*(k^1,\dots,k^n)$.
%\end{Prop}
%\begin{proof}
%  Since $U_j^t$ is Lipschitz-continuous in $k^j$ up to one point $\bar k$ in which it has a jump, one can simply calculate the derivative of the expectation by taking the expectation of the derivative plus the influence of the jump-point which produces the density at that point times what one would expect from the derivative.
%\end{proof}

We can now formulate the main result of this section. Basically, we try to optimise the choice of $k$ under the constraint that all $k^j$ have to be equal. This does not allow to optimise for every individual but we quantify that each individual cannot improve by much if every individual has a small wealth in the scheme compared to the total wealth of the scheme.
\begin{Thm}\label{t:diagonal}
  Define $z(k) := -\frac{2\theta_{t-1}+k\sum_{i=1}^n\eta^i_{t-1}}{2\theta_{t-1}+\sum_{i=1}^n\eta^i_{t-1}}$ for $k\in [0,1]$ and 
   $$ N(c,k) := \E[h(c,k)-\alpha h(c,k)] $$
  where $h(c,k) := \rho_t-\frac14(\rho_t-k)^++\frac12(-k-\rho_t)^+\one_{\{\rho_t>c\}}$ for $k\in[0,1]$ and $c\in[-1,0]$. For a given value of $c\in [-1,0]$ we denote the maximiser of $ N(c,\dots)$ by $k(c)$.
   
Assume that there is $\bar k\in[0,1]$ such that $\bar k = k(z(\bar k))$. 
  
  Then choosing the barrier $k$, for each individual at time $t-1$ is near to the optimal in the sense that changing the barrier $k^j$ for the $j$-th individual does not improve its performance by more than 
   $$ \|f\|_{\infty}\frac{(1/2 + \alpha)\eta^j_{t-1}}{2\theta_{t-1} +\sum_{i=1}^n \eta^i_{t-1}} $$  
  where $\|f\|_{\infty}$ denotes the maximum of the continuous density of $\rho_t$.
\end{Thm}
\begin{proof}
  %We take the view of the $j$-th individual and 
  We choose the barriers $k^i = \bar k$ for any other individual $i\neq j$. For the $j$-th individual we are supposed to maximise the function
   $$ \E[U_t^j-\alpha(U_t^j)^2] $$
  over the possible values of $k^j\in [0,1]$ and its maximiser is denoted by $\bar k^j$.  We simply write $z^*(k^j)$ when we mean $z^*(k^1,\dots,k^n)$ as a function of $k^j$ and the other $k^i=\bar k$. Define $c:=z(\bar k)=z^*(\bar k)$. $\bar k$ is the maximiser of the function $ N(c,\cdot) $ and  
\begin{align*} 
|N(c,k^j) &- \E[U_t^j-\alpha(U_t^j)^2]| \\&\le \E\bigg[\frac12 \one_{\{\rho_t<-k^j,\rho_t\in[c,z^*]\}}+\alpha (|U^t_j|+|N(c,k^j)|)\frac12\one_{\{\rho_t<-k^j,\rho_t\in[c,z^*(\bar k^j)]\}}\bigg].
\end{align*}
Since $U,N$ are bounded by $1$ on $\{\rho_t<0\}$ we find
     $$ |N(c,k^j) - \E[U_t^j-\alpha(U_t^j)^2]| \leq (1/2+\alpha) P(\rho_t\in[c,z^*(\bar k^j)]). $$
    Remark \ref{R:z star} yields $|c-z^*(\bar k^j)| \leq |\bar k^j-\bar k| \frac{\eta^j_{t-1}}{2\theta_{t-1}+\sum_{i=1}^n\eta^i_{t-1}} \leq \frac{\eta^j_{t-1}}{2\theta_{t-1}+\sum_{i=1}^n\eta^i_{t-1}}$ and the result follows.
\end{proof}
The theorem suggests a simple algorithm to find a nearly optimal choice, namely to choose a sequence $\bar k_n$ and $c_n$ recursively via $\bar k_0 = 1$, $c_0 = 0$ (or any other starting values) and define recursively
 \begin{align*}
    c_{n+1} &:= z(\bar k_n), \\
    \bar k_{n+1} &:= k(c_{n+1})
 \end{align*}
for any $n\in\mathbb N$. The value $c_n$ is the threshold where it is expected that such a downfall of the underlying fund makes it impossible to cover all the losses from the individual accounts and $\bar k_n$ the barrier chosen for all the individuals.\footnote{Also, using the mean field game theory, where each individual has contributed an actually negligible amount compared to the whole collective, optimisation of the barrier in this particular case is roughly the same as ignoring the possibility of the collective fund to become empty. There, due to the fact that the value functions for all individuals are equal, the optimality can only be attained by choosing the same barrier.}
\subsection{Using a redistribution index if the collective account does not have sufficient number of units}
Another possibility to handle the situation of insufficient number of units in the collective account is to use the redistribution index. 
If the collective account does not have enough units, the policies with a deficit can, for instance, claim a number of units corresponding to their redistribution index. It means the individual account has the value 
\begin{align*}
&V_t=\gamma\pi_{\rm ind} +\eta_{t-1}H_t-\frac14 V_{t-1}\Big(\frac{H_t}{H_{t-1}}-1-k\Big)^+\nonumber
\\&\quad {}+ \frac12 V_{t-1}\Big(1-k-\frac{H_t}{H_{t-1}}\Big)^+\one_{\big[2\theta_{t-1}\frac{H_t}{H_{t-1}}> \s_{i=1}^n\eta^i_{t-1}\big(1-k^i-\frac{H_t}{H_{t-1}}\big)^+\big]}
\\&\quad {}+  \min\bigg\{J_{t-1}\theta_{t-1}H_t,\frac12 V_{t-1}\Big(1-k-\frac{H_t}{H_{t-1}}\Big)^+\bigg\}
\\&\quad\quad  {}\times \one_{\big[2\theta_{t-1}\frac{H_t}{H_{t-1}}\le \s_{i=1}^n\eta^i_{t-1}\big(1-k^i-\frac{H_t}{H_{t-1}}\big)^+\big]},
\end{align*}
where $k^j=k$. In Section \ref{sec2}, we prove that the optimal corridor boundary for the return, $k$, is the same for all contracts in the pool of contributors. It means, if the fund go down all individual accounts will produce claims simultaneously. However, the claim sizes depend on the number of shares in the individual accounts and differ from contract to contract. Therefore, some contracts might produce claims smaller than the number of units in the collective account corresponding to their redistribution index and vice versa. If the regulation requirements allow to entirely empty the collective account, the following recursive procedure can be applied, see also Figure \ref{fig:12}:
\begin{figure}[t]
\centering
\begin{minipage}[h]{\textwidth}
\includegraphics[scale=0.45, bb = 20 -150 200 380]{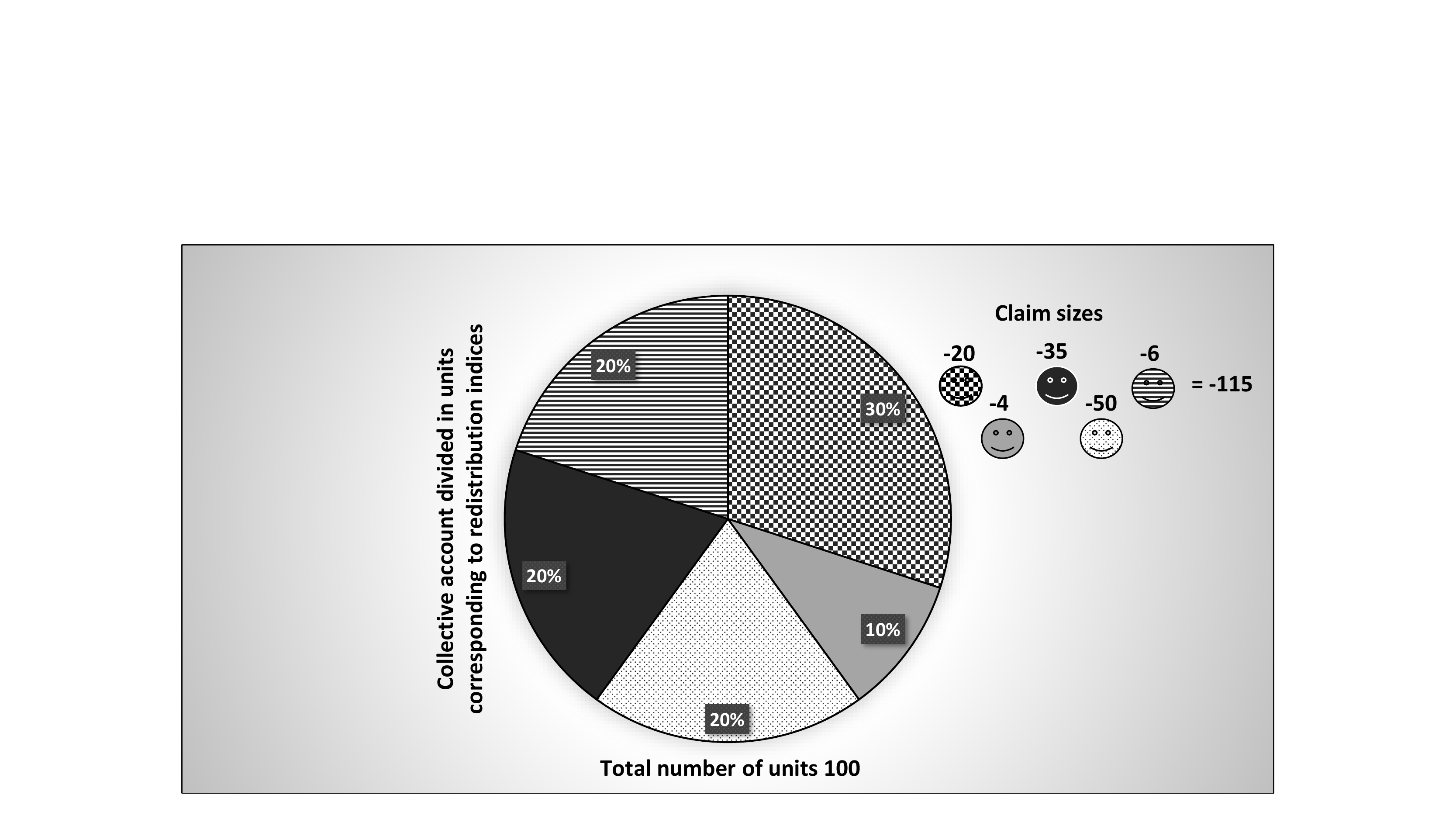}
\end{minipage}
\begin{minipage}[h]{\textwidth}
\includegraphics[scale=0.45, bb = 20 150 200 380]{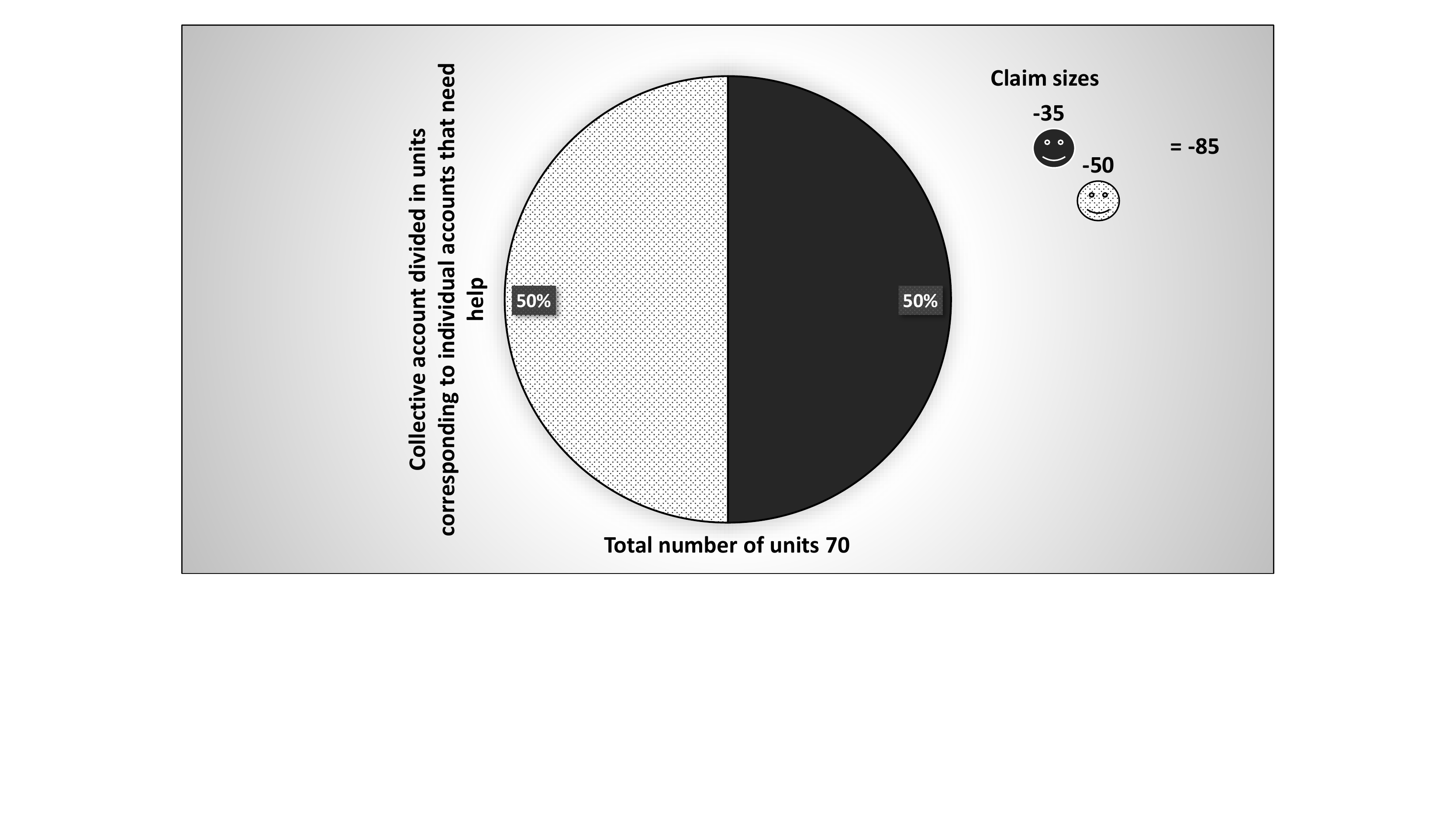}
\caption{Redistribution of the shares from the collective account by a recursive procedure. \label{fig:12}}
\end{minipage}  
\end{figure}
\begin{itemize}
\item Settle all individual claims that are below their redistribution part. \\In Figure \ref{fig:12} the claims amounting to $4$, $6$ and $20$ have the redistribution indices $0.1$ (yielding 10 shares), $0.2$ (yielding $20$ shares) and $0.3$ (yielding $30$ shares) respectively. It means these claims can be settled immediately. 
\item Adjust the redistribution indices of the remaining claims to the new number of claims and settle those that are now below their redistribution part.
\\In Figure \ref{fig:12}, after settling ``small'' claims in the first step, the collective account has $100-4-6-20=70$ shares on its disposal. The redistribution indices of the remaining two claims, amounting to $35$ and $50$ shares, equalled to $0.2$ in the first step and should be adjusted due to the new claim number of $2$. Therefore, the new redistribution indices are given by $0.5$ yielding $35$ shares. Thus, one claim can be completely covered. 
\item Proceed until all remaining claims exceed their redistribution part and eventually settle them. 
\\In our example, the collective account has now $35$ shares. The new redistribution index of the claim amounting to $50$ is now $1$. This contract gets just $35$ shares from the collective account which is now empty.
\end{itemize}
The above procedure serves just as an example and targets to showcase a possibility to handle the individual claims. Therefore, the presented numbers cannot be considered as realistic quantities. 
Also, it should be noted that redistributing all shares from the collective account between individual accounts might leave the next retiring cohort with small amounts of capital resulting from the collective account compared to the amount of premia they paid in if the size of the collective is not large enough. This would clearly violate the concept of fairness and require intergenerational smoothing mechanisms. 
\\On the one hand, the procedure of getting help from the collective account is a question of product design but on the other hand, it should also be in line with the regulations in place answering the intergenerational fairness and sustainability requirements. 
\\The above described recursion could also be applied on the returns of the collective account so that the main capital remains untouched. However, this procedure will contradict the primary mission of the collective account - to serve as a backup for the individual accounts. 
\medskip
\\Concerning the mathematical implementation of the scenarios described above, the method is similar to the one described in Section \ref{sec2}. Neither the value function nor the optimal strategy can be calculated explicitly. 
\section{Redistribution Index \medskip\\\large The individual share on the collective account -- A theoretical approach \label{reindex}}
%\section{Index -- Measuring individual shares on the collective account -- A theoretical approach \label{reindex}}
In this section we discuss mechanisms for measuring the share each individual has on the collective account.
\\We will work in discrete time $\mathbb T=\{0,\dots,N\}$ for some $N\in\mathbb N$. There is a finite number $K$ of individuals which share a collective account $C$. At each time $t\in\mathbb T$ the $j$th individual contributes an amount $J^j_t$. By $C_t$ we mean the collective amount at time $t$ before individual contribution and $C_{t+} := C_t + \sum_{j=1}^K J^j_t$ denotes the collective amount after individual contribution.
\\
To avoid ambiguity we assume that $J^1_0 > 0$, i.e.\ the first individual contributes at time zero and we assume that $C_0 = 0$, i.e.\ there is no money in the collective account prior to any contribution.
\\
The total amount $C_t$ of the collective account belongs in parts to the individuals, each individual owns a fraction $\rho^t_j \in \mathbb R$ on the collective account and $\sum_{j=1}^K \rho_t^j = 1$, i.e.\ the absolute share of the $j$th individual is $\rho_t^j C_t$ and the fraction after contribution is $\rho_{t+}^j$. ($\rho^j_0$ is meaningless because $C_0=0$).
\\
We now introduce several rules which seem natural to impose.
\begin{enumerate}
  \item[(Cont.)] {\em Contribution rule:} If the $j$th individual adds the amount $J_t^j$ at time $t$, then its absolute share equals its prior absolute share plus the contribution, i.e.\ 
   $$ \rho_t^jC_t + J_t^j  = \rho_{t+}^jC_{t+},\quad t=0,\dots,N. $$
  \item[(Fix)] {\em Returns do not change the relative share:} $\rho_{t+1}^j = \rho_{t+}^j$ for any $t=0,\dots,N-1$, $j=1,\dots,K$, i.e.\ the relative shares stay fixed during times when no contribution are made. (With this rule we will sometimes write $\rho_{N+1}$ instead of $\rho_{N+}$.)
  \item[(Mon.)] {\em More contribution means more share:} If the $j$th total individual contribution at any time prior to some time point $t$ is higher than those of the $k$-th total individual contribution, then the $j$th individual has a higher share at time $t$ than the $k$-th individual, i.e.
   $$ \left(\forall s=0,\dots,t :    \sum_{n=0}^s J^j_n \geq \sum_{n=0}^s J^k_n\right) \quad\Rightarrow\quad \rho_{t+}^j \geq \rho_{t+}^k. $$
   \item[(Add)] {\em New policyholders which are added do not change the relative share of the existing policyholders:} If an additional individual is added at some point $t$ with no contributions strictly before $t$ and some contribution $J_t^{K+1}>0$ at time $t$, then its share at time $t$ is strictly positive and the relative shares $\rho_{t+}^j/\sum_{k=1}^K\rho_{t+}^k$ of the other individuals is the same with and without the introduction of the new contributor.
  \item[(Lin.)] {\em Linearity of the contribution:} If the $j$th individual contributed $x$-times as much as the $k$-th individual at time $t$, then its absolute share increases $x$-times as much, i.e.
   $$ J^j_t = x J^k_t \quad \Rightarrow \quad \left(\rho_{t+}^jC_{t+}-\rho_t^jC_t = x(\rho_{t+}^kC_{t+}-\rho_t^kC_t)\right).$$
\end{enumerate}
 \begin{Rem}
The contribution rule (Cont.) simply means that someone who contributes $x$ owes $x$ more from the collective amount at that given time. Rule (Fix) means that there are no changes to the relative shares if the collective amount gains of looses value, e.g.\ if someone owns $10\%$ of a building and the building gains or looses value due to external factors, then the $10\%$ share remains fixed. The monotonicity rule (Mon.) means that someone who has contributed more at any time up to a fixed time $t$ also owes more than someone who has contributed less. Rule (Add) means that the relative distribution of $K$ contributors is unaffected by an additional contributor. The linearity rule (Lin.) means that at each fixed time the absolute share increases linearly depending on a factor which is the same for every individual but this factor may depend on the number of individuals or the time.
 \end{Rem}
We believe that the (Fix) rule is very natural and violating it means that there is some redistribution mechanism between the participants even if none makes any contribution. Redistributing for no reasons seems to be unfair for us and we will always assume that (Fix) is in place.
\\
The contribution rule (Cont.) together with the (Fix) rule does in fact determine the structure of the wealth distribution of the collective amount completely. Together, they imply a unique mechanism for the shares which does satisfy the rules (Add) and (Lin.) but can fail the monotonicity rule (Mon.). The latter fails if the collective amount $C$ falls between time steps, i.e.\ $C_{t+1}-C_{t+} <0$ for some $t$. We do not claim originality of the next statement which is known due to its triviality but might not have been recorded somewhere.
\begin{Prop}\label{p:share}
We assume that (Fix) and (Cont.) hold. Then (Add) and (Lin.) hold. %$C$ is increasing in the sense that $C_{t+1}-C_{t+} \geq 0$ for any $t=0,\dots,N-1$ if and only if (Mon) holds.
\\The dynamics of $\rho^j$ are uniquely determined. Moreover, the dynamics of $\rho^j$ can be described more statically in the following way:
\\
We denote by $I^j_t$ the index of the $j$th person at time $t$ which is defined by
     $$ I^j_t := 0 $$
     for any time $t$ strictly before the $j$th individual contributes to the collective, 
     \begin{align*}
        I_1^j &:= \frac{J_0^j}{J_0^1} 100
     \end{align*}
        for the index of the $j$th individual at time $1$,
      \begin{align*}
         I_{t+1}^j &:= I_{t}^j + \frac{J_{t}^j}{C_t}\sum_{l=1}^KI^l_{t}
      \end{align*}
      for any $t=1,\dots,N$. The relative share, resp.\ the absolute share of the $j$th individual at time $t=1,\dots,N+1$ is
       $$ \rho_{t}^j := \frac{I_{t}^j}{\sum_{l=1}^KI_{t}^l}, \quad\quad \rho_{t}^jC_{t} = C_{t} \frac{I_t^j}{\sum_{k}I^k_{t}}. $$
 \end{Prop}
 \begin{proof}
From the contribution rule (cont.) together with the (Fix) rule we find that the new relative contribution is given by
    $$ \rho_{t+1}^j = \frac{\rho_t^jC_t+J_t^j}{C_t+\sum_{j=1}^KJ_t^j} $$
    for any $t\in\mathbb T$, $j=1,\dots,K$. This fully determines the relative share for each time $t=1,\dots,N$ because by the above formula one has
    $$ \rho_1^j = \frac{J_1^j}{\sum_{j=1}^KJ_1^j}. $$
It is straightforward to verify the linearity (Lin.) and the addition (Add) rule.
\\Now, using the indices above and $\rho_t^j := \frac{I_t^j}{\sum_{l=1}^KI_t^k}$ yields
     $$ \rho_{t+1}^j = \frac{I_{t}^j + \frac{J_{t}^j}{C_t}\sum^K_{l=1}I^l_{t}}{\sum_{l=1}^K\left(I_{t}^l+\frac{J_{t}^l}{C_t}\sum^K_{l=m}I^m_{t}\right)} = \frac{\rho_t^j + J_{t}^j/C_t}{1+\sum_{l=1}^KJ_{t}^l/C_t} = \frac{\rho_t^jC_t+J_{t}^j}{C_t+\sum_{l=1}^KJ_{t}^l}$$
and thus, this is the same $\rho$ as before. This shows that the rule described by $I$ yields the unique rule to determine the relative shares (given that (Fix) and (Cont.) are imposed).
 \end{proof}
\begin{Rem}
The number $100$ appearing in Proposition \ref{p:share} is arbitrary and can be replaced by any non-zero number.
\end{Rem} 
It turns out that (Fix) and (Cont.), which describe the relative share at each point uniquely, do not imply the monotonicity rule (Mon.).
\begin{Bsp}
We will consider the time set $\mathbb T=\{0,1\}$ and $K=2$ individuals. The first individual decides to pay $100$ at time $0$ while the second decides to pay $80$ at time $1$. The collective amount is assumed to fall by $25\%$ from time zero to $1$. We find
     $$ C_0 = 0,\quad  C_{0+} = 100,\quad C_1 = 75,\quad C_{1+} = 155. $$
From this we find the relative shares
     $$ \rho_1 = (100\%,0), \quad  \rho_2 = \left(\frac{75}{155}, \frac{80}{155}\right) $$
or in terms of the index $I$, rounded with two digit precision after the dot,
     $$ I_1 = (100,0), \quad I_2 = (100,106.67) $$
\end{Bsp} 
In the above example, we can see that the second individual has a higher relative share than the first despite his smaller and later contribution. For an investment which can be precisely evaluated at every point of time this makes sense as the second contributor was only contributing after the (disastrous) $25\%$ downfall. For an investment which cannot be evaluated precisely at every point of time this rule might be less favourable and the monotonicity rule could be what is wanted instead. This should especially be used when in a real situation the precise timing of a downfall cannot be determined. 
\medskip
\\Our result on possible wealth sharing agreement under the monotonicity rule uses again an artificial index. This time, however, the index does not depend on the underlying value and can be computed even if it is unknown. It also shows that there is some sort of ``artificial interest rate'' which is added to the index. The artificial index can theoretically be time-dependent, individual dependent and scenario dependent but we believe that in practice a fixed value (e.g.\ expected return of the investment) would be chosen instead.
\begin{Prop}\label{prop:Mon}
It is possible that (Fix), (Add), (Lin.) and (Mon.) hold at the same time (for this it does not matter how $C_{t+1}-C_{t+}$ behaves.)
\\We now assume that (Fix), (Add), (Lin.) and (Mon.) hold, that there is a function
      $$ F:\mathbb T\times \mathbb R^K\times\mathbb R\times \mathbb R\times \mathbb R \rightarrow \mathbb R$$
with
      $$ \rho_{t+1}^j = F\left(t,\rho_t, J^j_t, \sum_{j=1}^K J^j_t, C_t\right). $$
Also, assume that $F$ is strictly increasing in its third variable if the other variables are fixed. Then there are factors $a^ j_t\geq 0$ for $t=1,\dots,N$, $j=1,\dots,K$ and one can define indices $I^j_t$ via:
     $$ I^j_t := 0 $$
for any time $t$ strictly before the $j$th individual contributes to the collective. 
\begin{align*}
        I_1^j &:= J_0^j
\end{align*}
the $j$th individual index at time $1$ and
      \begin{align*}
         I_{t+1}^j &:= I_{t}^j(1+a^j_t) + J_t^j
      \end{align*}
and the relative share, resp.\ the absolute share of the $j$th individual at time $t$ is
       $$ \rho_{t}^j := \frac{I_t^j}{\sum_{l=1}^KI^l_{t}}, \quad\quad \rho_{t}^jC_{t} = C_{t} \frac{I_t^j}{\sum_{l=1}^KI^l_{t}}. $$
 \end{Prop} 
 \begin{proof}
The rule implied by the index $I$ and $\rho$ as given by the index obviously satisfies the four rules. This shows that they can hold at the same time.
\\Now we assume that $\rho_t^j$ is given such that the four rules hold and define
     $$ I_0 := 0,\quad I_1 := J_0. $$
By the (Add) and (Lin) rule we find that a newly introduced contributor at time $t=1,\dots,N$ has an affine impact with its contribution on its share, i.e.\ $\rho_{t+1}^{K+1} = \beta_t + \gamma_t J_t^{K+1}$ for some constants $\beta_t,\gamma_t$ which depend on $C_t, C_{t+1}$ and the relative shares of the other contributors. Thus we find that
     $$ F(t,\rho_t,J_t,C_t) = G(t,\rho_t,C_t) + H(t,\rho_t,C_t) J^j_t $$
for some functions $G$, $H$ not depending on $J$. The monotonicity rule and the strictly increasing assumption on $F$ yield that $G \geq H> 0$. We define recursively
    \begin{align*}
       a_t^j := \frac{G(t,\rho_t,C_t)}{H(t,\rho_t,C_t)I_t^j} - 1 \geq 0, \\
       I_{t+1}^j := I_t^j(1+a_t^j) + J_t^j
\end{align*}   
and simply observe that by induction we have
      $$ \rho_t^j = \frac{I_t^j}{\sum_{j=1}^K I_t^j}. $$
\end{proof}
Proposition \ref{p:share} identifies the unique rule from the two rules (Fix) and (Cont.) which can be implemented very easily and it does not need any modelling assumptions on $C_t$ (i.e.\ how $C_{t+1}$ looks like given $C_{t+}$ and the scenario.) However, the rule introduces no penalty for late contribution other than that there is no participation on earlier gains (and losses). If risk is associated with $C$ and policyholders are risk-averse, then they tend to make late investments. For early contributions to become more appealing it seems natural to use the (Mon.) rule. Since (Fix.), (Add.) and (Lin.) are very natural rules one is obliged to use a setup as in Proposition \ref{prop:Mon}.
\section{Conclusions}
The life insurance companies are continuously creating new products to cope up not only with longevity but also with the period of protracted low interest rates. It is well-known that low interest rates affect investment opportunities and, in particular, have a significant adverse effect on insurers whose liabilities includes some benefit promises such as guarantees.\\
With the aim of offering an adequate level of benefits to the policyholders and at the same time preserving the long-term solvency of the plan, this paper analyses a new pension design applied to the accumulation phase from a mathematical point of view. Under the proposed design, we seek to maximise the accumulated capital at retirement by investing the premia into two funds: an individual and a collective. The collective fund acts as a buffer where some units are transferred to (from) the individual account when the performance of the individual fund is below (above) a particular barrier.\\
We prove that, in the case of symmetric boundaries for the corridor $[-k,k]$ and if the collective account never ruins, the optimal $k$ for the exchange of units between the individual and the collective account is given either by the lowest barrier allowed by the profitability condition or by $1$. If the barriers are asymmetric, we might have cases where no satisfactory results are obtained because the profitability condition is not fulfilled. 
\\In order to incorporate the possible risk into the target functional to maximise, we include a penalty function given by the expected realised volatility of the fund. We also show that in some occasions the maximum might not be unique. However, as the individual account is the main basis to calculate the initial pension, we choose the barrier that optimises the expected value of the individual account for the policyholder.\\
Due to the lack of analysis in both academia research and practice, this paper also analyses the possibility of not having enough units to be transferred to the individual account. The first solution analysed in this paper is the no-transfer of any units if the number of units is not enough to cover all claims. Due to the cross dependence among all individuals we make a precise error analysis where the same barrier is used by every policyholder. The same barrier is a suboptimal strategy but the explicit solution of an optimal $k$ would require lot of computation and the improvement for the policyholder would be negligible, i.e. the maximised amount at retirement would barely increase. Secondly, we describe a redistribution index that could cover some of the deficit of the individual claims and could be applied through a recursive procedure. In this paper, we also discuss the necessary and sufficient conditions for the redistribution index design.\\
This paper presents an innovative and attractive way to smooth the volatility of the fund in the accumulation phase. Hence, the proposed product design should be beneficial to both the life insurers $-$ as there are no benefit promises $-$ and policyholders $-$ as the amount of accumulated capital is more secure than in the case of risky investments and much higher than in the case of non-risky investments. \\
Finally, based on the model presented, at least three important directions for future research can be identified. First, another challenge in the accumulation phase is the maximisation of the retirement capital through an optimal splitting strategy of the premia into the two funds, i.e. individual and collective. Another avenue for future research would be to explore the redistribution index so that it ensures the intergenerational fairness among the members' plan. Third, it would be interesting to set up bounds for the pension amount during the retirement phase so that the retirees have a stable benefit level over time.
\subsection*{Acknowledgements}
Mar\'ia del Carmen Boado-Penas is grateful for the financial assistance received from the Spanish Ministry of the Economy and Competitiveness [project ECO2015-65826-P]. \medskip
\\The research of Julia Eisenberg was funded by the Austrian Science Fund (FWF), Project number V 603-N35.
\bibliographystyle{plain}

\begin{thebibliography}{}
%
\bibitem{boado} Boado-Penas, C., Eisenberg, J., Helmert, A. and Kr\"uhner, P. (2019), A new approach for satisfactory pensions with no guarantees, European Actuarial Journal. Forthcoming.
%
\bibitem{bohnert} Bohnert,A., Born, P. and Gatzert, N. (2014) Dynamic hybrid products in life insurance: Assessing the policyholders' viewpoint, Insurance: Mathematics and Economics, 59(1), 87--99.
%
\bibitem{bovenberg} Bovenberg, A.L. (2009) Dutch stand-alone collective pension schemes: The best of both worlds? In Z. Bodie, L.B. Siegel \& R.N. Sullivan (Eds.), The Future of Life-Cycle Saving and Investing: The Retirement Phase, 69--76. Charlottesville: CFA Institute.
%
\bibitem{dickson} Dickson, D.C.M. (2005) Insurance: risk and ruin, Cambridge University Press, New York.
%
\bibitem{gatzert} Gatzert, N. and Schmeiser, H. (2013) New life insurance financial products. In Handbook of Insurance, 1061--1095. Springer.
%
\bibitem{goek} Goecke, O. (2013) Pension saving schemes with return smoothing mechanism, Insurance: Mathematics and Economics, 53(3), 678--689. 
%
\bibitem{guillen} Guill\'en, M., Jorgensen, P.L. and Nielsen, J.P. (2006) Return smoothing mechanisms in life and pension insurance: Path-dependent contingent claims. Insurance: Mathematics and Economics, 38(2), 229--252.
%
\bibitem{ponds2} Ponds, E.H.M. and Van Riel, B. (2009) Sharing risk: the Netherlands' new approach to pensions, Journal of Pension Economics \& Finance, 8(1), 91--105.
%
\bibitem{binsber} Van Binsbergen, J.H, Broeders, D., De Jong, M. and Koijen, R.S.J. (2014) Collective pension schemes and individual choice, Journal of Pension Economics \& Finance 13(2), 210--225.
\end{thebibliography}

\end{document}